\pdfoutput=1
\documentclass{llncs}
\usepackage{fullpage}
\usepackage{latexsym}
\usepackage[fleqn]{amsmath}
\usepackage{amssymb}
\usepackage{enumerate}
\usepackage{comment}
\usepackage{graphicx}
\usepackage{algorithm2e}
\usepackage{url}

\begin{document}
\title{Optimizing XML Compression}
\subtitle{(Extended Version)}
\author{Gregory Leighton \and Denilson Barbosa}
\institute{University of Alberta\\Edmonton, AB, Canada\\ \email{\{gleighto;denilson\}@cs.ualberta.ca}}

\maketitle

\begin{abstract}
The eXtensible Markup Language (XML) provides a powerful and flexible means of encoding and exchanging data. As it turns out, its main advantage as an encoding format (namely, its requirement that all open and close markup tags are present and properly balanced) yield also one of its main disadvantages: verbosity. XML-conscious compression techniques seek to overcome this drawback. Many of these techniques first separate XML structure from the document content, and then compress each independently. Further compression gains can be realized by identifying and compressing together document content that is highly similar, thereby amortizing the storage costs of auxiliary information required by the chosen compression algorithm. Additionally, the proper choice of compression algorithm is an important factor not only for the achievable compression gain, but also for access performance.  Hence, choosing a compression configuration that optimizes compression gain requires one to determine (1) a partitioning strategy for document content, and (2) the best available compression algorithm to apply to each set within this partition.  In this paper, we show that finding an optimal compression configuration with respect to compression gain is an \textbf{NP}-hard optimization problem. This problem remains intractable even if one considers a single compression algorithm for all content. We also describe an approximation algorithm for selecting a partitioning strategy for document content based on the branch-and-bound paradigm.
\end{abstract}

\section{Introduction}
\label{sec:introduction}

The \textit{eXtensible Markup Language (XML)} has become increasingly popular as a data encoding format. XML has many benefits, but one notable weakness: its verbosity, resulting from the high markup-to-content ratio imposed in large part by requiring every markup tag to be properly closed. The increasing size of XML datasets has motivated researchers to seek ways to reduce storage costs by applying compression techniques. Because XML is inherently a textual format, the naive solution is to apply a generic text compression scheme. However, such schemes are not aware of XML syntax, and therefore cannot easily exploit redundancies in the tree structure unambiguously induced by the proper nesting of markup tags inside the XML document (such as repeated subtrees), or even distinguish an element tag from a text segment. Thus, such a strategy severely hinders query processing, which is fundamentally based on traversing the structure of the document.

With such shortcomings in mind, many {\em XML-conscious} compression techniques have been proposed in recent years. Among them, {\em homomorphic approaches} to XML compression (e.g.,~\cite{Adiego05,Cheney01,Cheney05,Leighton05b,Min03,Tolani02}) preserve the original tree structure in the compressed representation by processing each node as it occurs during a pre-order traversal. {\em Permutation-based approaches} (e.g.,~\cite{XQueC07,Leighton05a,Liefke00,Maneth08,Skibinski08}) re-arrange the document before performing compression, in an attempt to group ``similar'' nodes together and therefore improve the achievable compression rate. A commonly used permutation strategy treats structure separately from content, and then applies a partitioning strategy to group content nodes into a series of {\em data containers}. However, there is an inherent tradeoff between the achievable compression rate and access performance: in general, better compression tends to occur by grouping large sets of nodes together before compression, yet such a strategy will often hurt access time by increasing the number of decompression operations needed to extract relevant document fragments. 

In this paper, we focus on the permutation-based approaches, and seek to determine the complexity of determining optimal strategies for {\em container grouping} and {\em compression algorithm selection} such that the resulting {\em compression configuration} maximizes the overall compression gain, while keeping compression and/or decompression time and compression model storage requirements within specified bounds. Arion et al~\cite{XQueC07} were the first to investigate (albeit informally) the tradeoff between compression rate and query performance, given a set of typical queries, a set of available compression algorithms, and a specific XML database as inputs. We consider a more general setting that captures the problem outlined in~\cite{XQueC07} as well as additional application domains, including data archiving and data exchange. We provide a complexity analysis indicating that the difficulty of selecting an optimal compression configuration is {\bf NP}-hard, and also describe an approximation algorithm based on a branch-and-bound technique that finds the optimal compression configuration within polynomial time (w.r.t. the document size and the number of available compression algorithms), with the choice of appropriate parameter values. 

The paper is structured as follows. Section~\ref{sec:preliminaries} provides preliminary definitions and a background into the problem. Section~\ref{sec:analysis} investigates the difficulty of choosing an optimal tradeoff between compression gain and query performance. Section~\ref{sec:algorithm} describes an approximation algorithm for choosing a near-optimal compression configuration, while Section~\ref{sec:conclusion} concludes the paper and outlines our future work.

\section{Preliminaries}
\label{sec:preliminaries}
\subsection{XML Data Model}
\label{subsec:XMLModel}
We recall that an XML document can be represented as a rooted, ordered, labeled tree (the {\em document tree}), in which the leaf nodes correspond to attribute values and text segments (document content), while the interior nodes represent attributes and elements (document structure). According to convention, we distinguish attribute names from elements by prepending the former with `@'. As an illustrative example application, we consider a social recommendation website, where users share their opinions of movies, music, etc. with other users.  Additionally, users assign a prestige to other users, allowing them to express their evaluation of the quality of those users' recommendations. User account data is stored as XML; Fig.~\ref{fig:docTree} shows a fragment of the document tree.  

Query languages for XML center around {\em path expressions}, which are used to specify subsets of nodes within the document tree.  The two most influential XML query languages are XPath~\cite{XPath07} and XQuery~\cite{XQuery07}. 

\begin{example}
\label{ex:XQuery}
For the example document tree of Fig.~\ref{fig:docTree}, the following XQuery returns the titles of movies rated at least 4.5 by users with a prestige ranking lower than 4.

\begin{verbatim}
let $movies := for $user in doc(``ratings.xml'')//user 
   where $user/prestige lt 4.0
   return $user/favorites/movies/movie
   
for $movie in $movies
	where $movie/rating ge 4.5
	return $movie/title
\end{verbatim}
This query returns {\tt <title>Smoke</title>}.
\end{example}

\begin{figure}[t]
\centering
\includegraphics[scale=0.54]{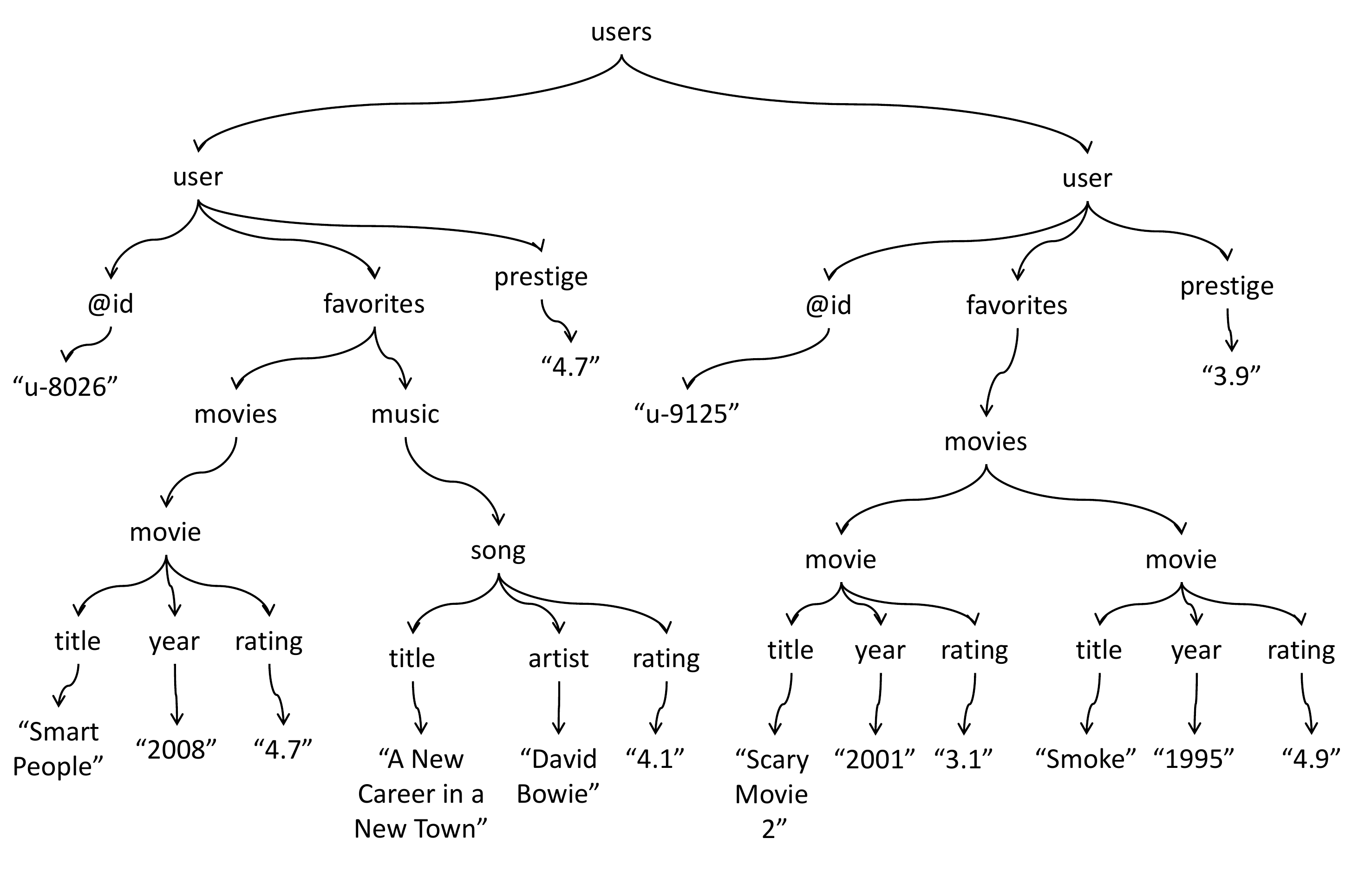}
\caption{Example XML document tree.}
\label{fig:docTree}
\end{figure}

\subsection{XML Compression}
\label{subsec:XMLCompression}

Permutation-based strategies for XML-conscious compression separately compress the document structure and text content.  The textual content is organized into containers, usually based on the path (or just the name) of the parent element. The intuition for doing so is that values belonging to different instances of the same element are likely to exhibit similarities that facilitate compression. Fig.~\ref{fig:defaultGrouping} shows the default path-based partitioning of the text content of the document tree in Fig.~\ref{fig:docTree}, in which data values belonging to each distinct element and attribute type stored in a separate container.  

\begin{figure}[t]
\centering
\includegraphics[scale=0.54]{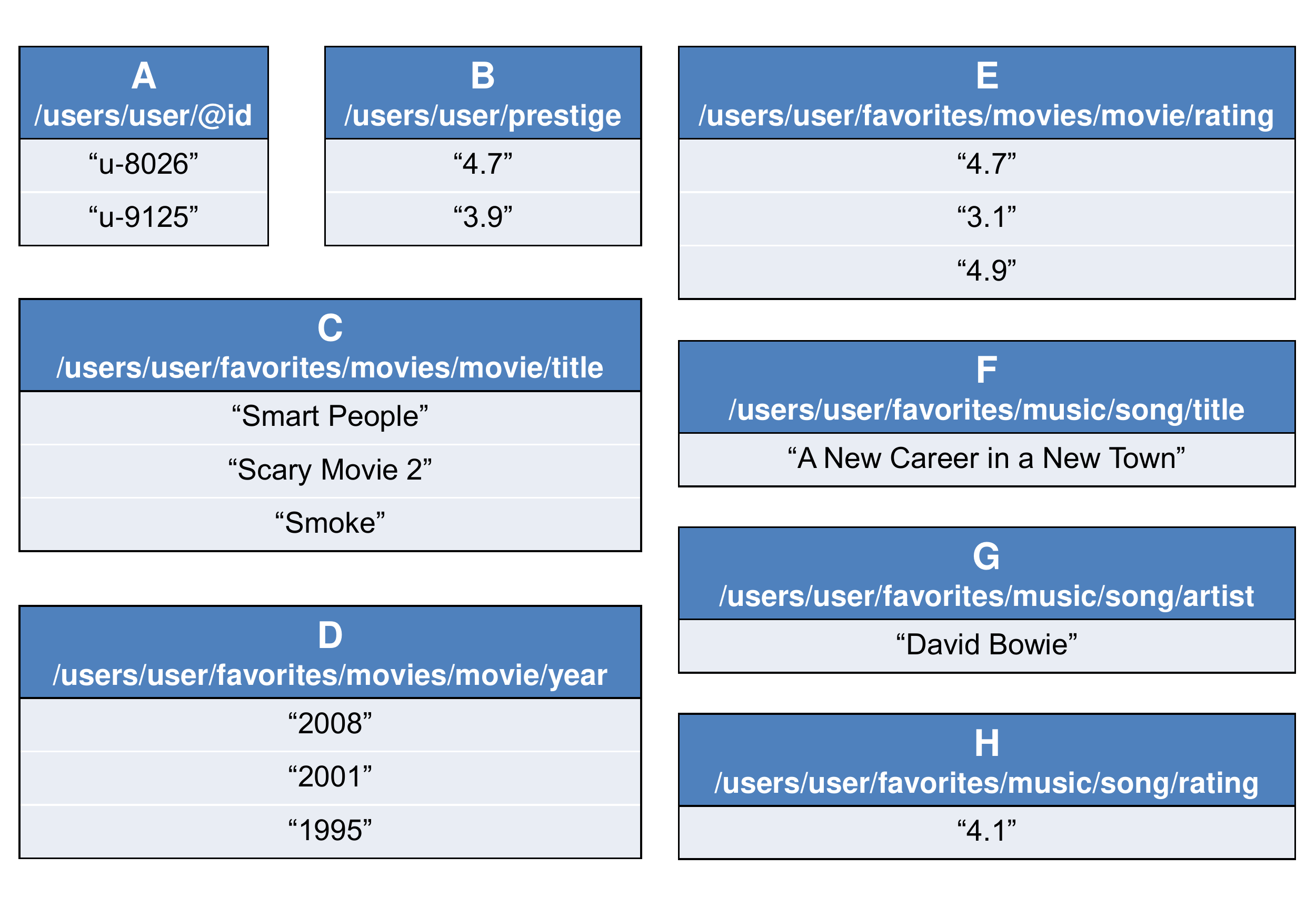}
\caption{A path-based partitioning of data values from the document of Fig.~\ref{fig:docTree}.}
\label{fig:defaultGrouping}
\end{figure}

Further compression gains can often be realized by generalizing the partitioning strategy to take into account additional factors, such as the data type of the content (e.g., integers, dates, and strings).  Grouping together multiple containers with high pairwise similarity allows the containers to share the same compression source model, reducing storage costs while simultaneously allowing more complex models over the longer sequence to be built.  Fig.~\ref{fig:chosenConfig} depicts a logical partitioning strategy that extends the default strategy from Fig.~\ref{fig:defaultGrouping}. Here, containers B, E, and H are grouped together, since user prestige, movie ratings, and song ratings are highly similar (i.e.,~they all consist of a real number value in the range [0.0, 5.0]).  Similarly, since the titles of movies and songs and artist names are all free-form text, it may prove beneficial to group together containers C, F, and G.

\begin{figure}[t]
\centering
\includegraphics[scale=0.54]{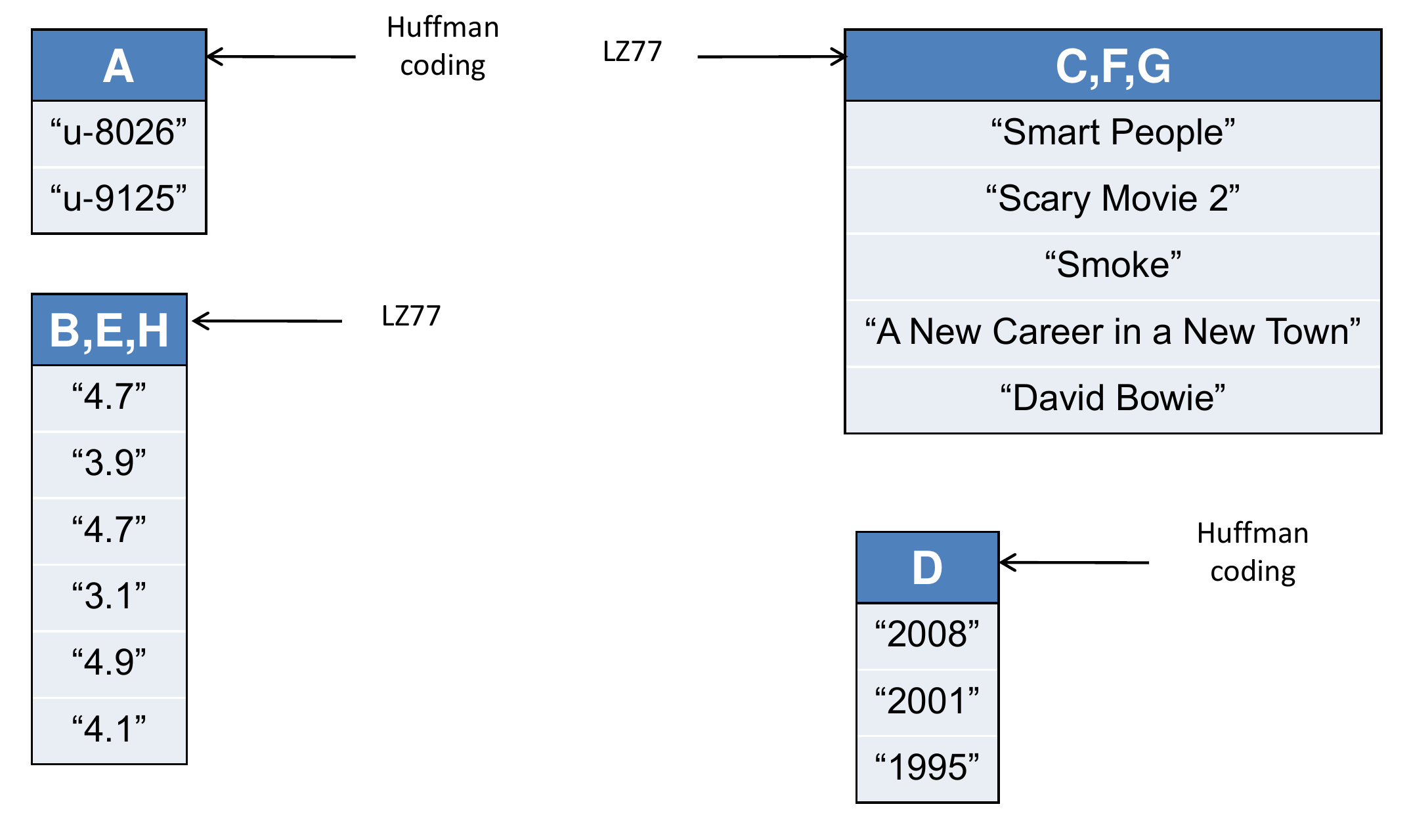}
\caption{A compression configuration for the document in Fig.~\ref{fig:docTree}.}
\label{fig:chosenConfig}
\end{figure}

Yet compression gain is often not the only criterion guiding the selection of a container grouping strategy. The choice of a partitioning strategy can also impact the efficiency of random access to nodes within the document tree.  In particular, query performance can be improved by choosing a partitioning strategy that places data segments involved in a common query within the same container subset.  Doing so can dramatically reduce the number of required decompression operations.  For Ex.~\ref{ex:XQuery}, a beneficial partitioning strategy might instead group together containers B, C, and E.

Proper algorithm selection is also an important factor to consider.  Greater compression can be realized by choosing a compression algorithm that is well-suited for the type of data values stored in a container subset.  Query performance is also impacted by the choice of compression algorithm, as the time required to carry out decompression adds to the query response time.  Fig.~\ref{fig:chosenConfig} additionally assigns a compression algorithm to each container subset (in this case, either LZ77 or Huffman coding).

Furthermore, certain compression algorithms allow classes of operations to be carried out without prior decompression; the choice of such an algorithm can therefore speed up query performance.  For example, using an order-preserving algorithm to compress user prestige and movie rating values would allow the comparisons in both {\tt where} clauses of the XQuery in Ex.~\ref{ex:XQuery} to be computed within the compressed domain, without requiring the decompression of each such value beforehand.

\subsection{XML Compression Measures}
\label{subsec:XMLCompressionMeasures}
We now consider the relevant measures used to evaluate solutions to XML compression problems, and discuss the inherent trade-offs between these values.

{\em Storage gain} measures the relative amount of space saved by applying a compression algorithm $a$ to a container $C$, denoted as $gain(C,a)$.  An effective measure must not only account for the size of the compressed representation of $C$; it must also consider the additional space required to store auxilary data structures constructed by the compression source model (e.g.,~for the Huffman algorithm, this would indicate the size of the generated tree; in dictionary-based compression schemes, it would represent the size of the dictionary).  It is calculated as

\begin{equation}
\label{eq:gain}
gain(C,a) = 1 - \frac{\mbox{compressed size of }C + \mbox{compression model size}}{\mbox{original size of }C}\enspace .
\end{equation} 

This measurement is also applicable to {\em sets} of containers; given a subset $S \subseteq \mathcal{C}$ and a compression algorithm $a$, $gain(S, a)$ is calculated by first concatenating the contents of each container in $S$, and then using the compressed and original sizes of this concatenated container, together with the storage costs of the generated compression model, in the above formula.

{\em Compression cost} and {\em decompression cost} measure, respectively, the time required to apply and reverse the compression process.  Both time measures are largely dependent on the contents of the container(s) being compressed, as well as the compression algorithm being employed.  Often, there is an inverse relationship between compression gain and compression/decompression time requirements, as the most effective compression schemes tend to heavily transform the original sequence, while less sophisticated schemes that allow fast compression and decompression typically fare worse in terms of compression performance.   By $comp(S,a)$ and $decomp(S, a)$, we denote, respectively, the time required to compress and decompress the contents of a container subset $S$ that has been compressed with algorithm $a$.

In the sequel, we make the assumption that all three measures can be calculated in polynomial time (with respect to the size of the input container subset).  Practical compression programs do feature a linear running time,  meaning that in theory, an exact calculation of each measure can be obtained by simply running the compressor on the input.  However, in many cases, the time and computational effort required to do this outweighs the benefit of obtaining a precise measurement (e.g., if the number of containers and/or compression algorithms is large), leading one to adopt faster and less expensive methods that provide reasonable estimates of compression gain, compression cost, and decompression cost. 

\paragraph*{Trade-offs between compression gain and decompression cost. }  Across application domains, the primary motivations for compressing XML differ.  These goals typically dictate a tradeoff between compression gain and compression and/or decompression cost.  We briefly examine the three most common use cases for XML compression.

\begin{itemize}
\item {\em Data archiving. }This encompasses any application where large volumes of XML-encoded data must be preserved, yet accessed infrequently (e.g., web server logs). Here, the fundamental goal is to conserve disk space by maximizing the achievable compression gain, while much less emphasis is placed on time costs.  This is because it is expected that each data set will only be compressed once, and only needs to be decompressed on the infrequent occasions that it is accessed.  
\item {\em Data exchange. }In this class of applications, XML-encoded data is exchanged between multiple parties, typically over a network.  XML documents are typically small in size and have a short life span (e.g.,~RPC-style web services, where an XML document only includes information related to a single service request or response. Once the document is processed, it is immediately discarded).  Here, the use of XML compression seeks mainly to improve application throughput and reduce bandwidth consumption by reducing the average size of transmitted messages, while not imposing excessive additional costs for compression and decompression (with the rationale being that extra time spent compressing and decompressing messages must not outweigh the performance benefits gained by employing compression in the first place). 
\item {\em Database applications. }In this scenario, the XML document is treated as a persistent data store over which queries are issued.  In many instances, it is possible to anticipate the types of queries which will be most commonly issued, which can be represented as a {\em query workload}. XML compression techniques are used primarily to improve query performance, by reducing the required number of disk reads/writes. As with any database application, query performance is paramount, and therefore minimizing decompression cost (particularly over the workload) becomes as important as maximizing the compression gain.  {\em Query-friendly} compression schemes, allowing certain queries to be carried out directly on the compressed representation, are highly desirable as they can reduce the necessary number of decompression operations.  Note that in our model, we can effectively capture the query performance for a given compression configuration over a particular workload within the decompression cost, as the latter measure is a function both of the chosen container partition strategy and the chosen compression algorithms for each container subset within the partition.
\end{itemize}

\section{Complexity Analysis of Compression Configuration Selection}
\label{sec:analysis}

We recall from the discussion in Sec.~\ref{subsec:XMLCompression} that our goal is to discover an optimal {\em compression configuration}, specifying both a partitioning strategy of the container set $\mathcal{C}$ and an assignment of a compression algorithm to each partition set. In this section, we demonstrate the {\bf NP}-hardness of this problem.

\begin{definition}
A {\em configuration} $\langle P, \alpha\rangle$ consists of a partition $P = \{S_1,\ldots,S_t\}$ of $\mathcal{C}$, and an algorithm assignment function $\alpha: P \rightarrow \mathcal{A}$ that assigns to each $S \in P$ a compression algorithm $a \in \mathcal{A}$. \qed
\end{definition} 

\begin{definition}
An instance of the optimization version of the {\em optimal compression configuration} problem consists of the following inputs:

\begin{itemize}
\item a set of available compression algorithms $\mathcal{A} = \{a_1,\ldots,a_q\}$;
\item a set of containers $\mathcal{C} = \{C_1,\ldots,C_x\}$; 
\item $gain: 2^{\mathcal{C}} \times \mathcal{A} \rightarrow \mathbb{Q}$, a function indicating the compression gain obtained when a specific compression algorithm in $\mathcal{A}$ is applied to a specific container subset in $2^{\mathcal{C}}$;
\item $comp: 2^{\mathcal{C}} \times \mathcal{A} \rightarrow \mathbb{Q}$, a function indicating the time cost associated with a compression of a specific container subset in $2^{\mathcal{C}}$ using a specific algorithm in $\mathcal{A}$;
\item $decomp: 2^{\mathcal{C}} \times \mathcal{A} \rightarrow \mathbb{Q}$, a function indicating the time cost associated with decompressing a specific container subset in $2^{\mathcal{C}}$ that has previously been compressed using a specific algorithm in $\mathcal{A}$;
\item $T_c$, an upper bound on total compression cost; and
\item $T_d$, an upper bound on total decompression cost.
\end{itemize}

The goal is to discover a configuration $\langle P, \alpha \rangle$ that maximizes $$\displaystyle \sum_{S \in P} gain(S, \alpha(S))$$ subject to 

\begin{eqnarray*}
\displaystyle \sum_{S \in P} comp(S, \alpha(S)) & \le & T_c \\
\end{eqnarray*}
and
\begin{eqnarray*}
\displaystyle \sum_{S \in P} decomp(S, \alpha(S)) & \le & T_d \enspace .
\end{eqnarray*}

In the decision version of the problem, there is an additional input $L \in \mathbb{Q}^{*}$ and a solver outputs ``yes'' if there exists a configuration $\langle P, \alpha \rangle$ such that $$\displaystyle \sum_{S \in P} gain(S, \alpha(S)) \ge L$$ subject to the given constraints, and ``no'' otherwise.
\end{definition}

\begin{theorem}
\label{thm:configurationSelection}
Selecting an optimal compression configuration is {\em {\bf NP}-hard}.
\end{theorem}

\begin{proof}
To show that optimal compression configuration is in {\bf NP}, note that it is possible to efficiently verify that a given configuration $\langle P, \alpha\rangle$ yields a ``yes'' answer by calculating the respective sums for the $gain$, $comp$, and $decomp$ values of each $(S,\alpha(S))$ pair and ensuring they obey the specified constraints for $L$, $T_c$, and $T_d$, respectively.  

{\bf NP}-completeness follows from a polynomial-time reduction from {\tt SUBSET SUM}, whose decision version is known to be {\bf NP}-complete~\cite{Garey79}.  Recall that an instance of the decision version of {\tt SUBSET SUM} consists of a set $Y = \{y_1,\ldots,y_r\}$ of elements, a function $s: Y \rightarrow \mathbb{Z}^{*}$ that assigns a score to each $Y$-element, and a value $B \in \mathbb{Z}^{+}$.  A solver outputs ``yes'' if there exists a subset $Y^{\prime} \subseteq Y$ such that $\sum_{y \in Y^{\prime}} s(y) = B$, and ``no'' otherwise.

The reduction proceeds as follows. From $Y = \{y_1,\ldots,y_r\}$, the container set $\mathcal{C} = \{C_1,\ldots,C_r\}$ is constructed.  $\mathcal{A}$ contains a single compression algorithm $a$; for a specific container subset $S \subseteq \mathcal{C}$, we set $gain(S,a) = comp(S,a) = decomp(S,a) = \sum_{C \in S} s(C)$.  Furthermore, $L = T_c = T_d = B$.  The only compression configurations yielding a ``yes'' answer correspond under this mapping to an instance of {\tt SUBSET SUM} which would yield a ``yes'' answer from the {\tt SUBSET SUM} solver.  

Since the decision version of the optimal compression configuration problem is {\bf NP}-complete, it follows that its optimization version is {\bf NP}-hard.\qed
\end{proof}

As a consequence of the preceeding proof, we may also infer the following about the variant of optimal compression configuration selection where the same compression algorithm is applied to all container groups.

\begin{corollary}
Selection of an optimal compression configuration remains {\em {\bf NP}-hard} when $|\mathcal{A}| = 1$.
\end{corollary}

This indicates that the ``hardness'' of the overall problem is not caused by algorithm selection, rather it is due to the difficulty of determining an optimal container partitioning strategy.  Indeed, we can determine an optimal algorithm selection for a specified container partitioning strategy in $O(|\mathcal{A}|\cdot |\mathcal{C}|)$ time by simply testing each available algorithm on each container subset.

\section{An Approximation Algorithm for Compression Configuration Selection}
\label{sec:algorithm}
In this section, we describe an approximation algorithm for selecting an optimal compression configuration. Throughout the discussion, we use the term {\em container subset} to refer to one or more containers which have been grouped together, and {\em grouping} to indicate a set of container subsets.  A grouping which covers {\em all} containers (i.e.,~assigns each container to exactly one container subset) is referred to as a {\em partitioning strategy}. 

In the first phase of the approximation algorithm (Sec.~\ref{subsec:BandB}), a branch-and-bound strategy is used to select a set of {\em candidate partitioning strategies}: ~a set of partitioning strategies which are estimated to be highly compressible.  In the second phase (Sec.~\ref{subsec:optConfig}), this set of partitioning strategies is tested against the set of available compression algorithms to determine the single compression configuration that yields the highest compression gain, while obeying the specified upper bounds on compression and decompression costs. 

We recall from Eq.~\ref{eq:gain} that computing the compression gain of a container subset $S$ is based on two additional measures: the size of the compressed representation of $S$, and the additional storage cost incurred by the generated compression model.  In the remainder of this section, we first describe how container compressibility and storage costs are estimated, and then discuss how these estimates are used in the computation of compression gains. We then detail both phases of the approximation algorithm. 

\subsection{Estimating Compressibility}
In his classic paper~\cite{Shannon48}, Shannon proved that the entropy rate $r$ of a stationary stochastic process represents a bound on lossless compression of any message emitted by that process.  For a set of random variables ${X_1,\ldots,X_n}$ whose values are drawn from a finite alphabet $\mathcal{X}$, 

$$r = \lim_{n\to\infty} H(X_n|X_{n-1},X_{n-2},\ldots ,X_{1})$$
where $H(X_n|X_{n-1},X_{n-2},\ldots,X_{1})$ denotes the conditional entropy of $X_n$ when the values of \linebreak $X_{n-1},X_{n-2},\ldots,X_1$ have been witnessed. Essentially, each $X_i$ represents a separate message from the same source; as one receives more and more of these messages (i.e.,~$n$ approaches infinity) they have more of a history to base the entropy estimate on, and hence the estimate will approach the true entropy value more closely.

While Shannon's entropy rate does provide a theoretical lower bound on compressibility, it proves to be impractical for our setting.  This is because the entropy rate is an asymptotic measure calculated by increasing the compression block size $n$ to infinity, while this is clearly impossible to do when one's knowledge of the source consists of a single string of finite length.  In other words, a single finite length string often does not provide enough opportunity to ``learn'' the source well enough to achieve an accurate measure of the true entropy.

We instead turn to Lempel and Ziv's method for calculating string complexity~\cite{Lempel76}. In this approach, which we refer to as LZ76, the input string $x$ is parsed once from left-to-right, and a set of phrases $\mathcal{P}_{x}$ are recursively built and added to a dictionary.  Once parsing has been completed, the complexity of $x$ is 

\begin{equation}
\label{eq:LZ}
C_{\mathrm{LZ}}(x) = \frac{|\mathcal{P}_{x}|}{|x|}\enspace ,
\end{equation}
the ratio of phrases per character.  Lempel and Ziv showed that this approach yields an approximation ratio of $\frac{n}{\log n}$ to Shannon's entropy rate.  

We now describe the parsing process of LZ76 in greater detail.

\begin{enumerate}
\item Initialize the dictionary to be empty.
\item If the end of $x$ has been reached, terminate.  Otherwise, read the next character from $x$ and assign it to phrase $p$. If $p$ matches an existing entry in the dictionary, continue reading characters from $x$ and appending them to $p$ until $p$ no longer matches an existing dictionary entry.    
\item Assign $p$ the next available index position, and add both the index value and $p$ to the dictionary. Go to step 2.
\end{enumerate}

\begin{example}
\label{ex:LZparsing}
For a container subset $S$ with contents ``aaabc'', the first iteration constructs the phrase $\langle \text{a}\rangle$ (step 2) and adds it to the dictionary (step 3) at index position 1.  The second iteration reads the next character (`a') from $x$ and assigns it to $p$. Since $\langle \text{a}\rangle$ is in the dictionary, the next character is read from $x$ and appended to $p$ to form the phrase $\langle \text{aa} \rangle$, which does not appear in the dictionary. This phrase is added to the dictionary at index position 2. The following iterations construct phrases $\langle \text{b} \rangle$ and $\langle \text{c} \rangle$ and add them to the dictionary at index positions 3 and 4, respectively.  We then compute the complexity of $S$ as $C_{\mathrm{LZ}}(S) = 4/5 = 0.8$.   
\end{example}

\subsection{Estimating Storage Cost}
\label{subsec:storageCostModel}

Obtaining a complete picture of the achieved storage gain via compression requires one to take into account not only the size of the compressed data itself, but also the additional space required to store information about the compression model used.  Compression models typically consist of a mapping between the uncompressed symbol alphabet and the corresponding codewords assigned to each symbol by the compression algorithm. 

To compute the storage gain for a container subset, we simulate the cost of transmitting the dictionary using the coding strategy of LZ78~\cite{Ziv78} (recalling that LZ78 utilizes the parsing strategy of LZ76 in concert with a specific coding strategy for dictionary phrases). Each time a new phrase of length $l$ is constructed, two pieces of information are emitted to the compression stream: (1) a codeword $W$, representing the index position of the existing phrase $p$ of length $l-1$ that forms a prefix of the new phrase, and (2) the ``innovative'' character $c$ that is appended to $p$ to form the new phrase.  Since phrase indexing begins at 1, the highest index value for a dictionary with $t$ phrases will be $t$.  Using a fixed-length encoding, then, we can express each $W$ value using $\log_2(t)$ bits, requiring a total of $t \cdot \log_2(t)$ bits to encode all codewords.  Furthermore, a single character $c$ is emitted each time a new phrase is created, requiring an extra $8 \cdot t$ bits (here, we assume a text encoding that requires a single byte per character is in use; multibyte formats can be incorporated by replacing 8 with the number of bits per character used in the chosen encoding format).

\begin{definition}
The {\em storage cost} (expressed in bits) associated with a container subset $S$ is calculated as

\begin{equation}
\label{eq:storageCost}
storageCost(S) = t \cdot (8 + \log_2(t))
\end{equation} 
where $t$ is the total number of entries in the dictionary after an LZ76 parsing of $S$.

The storage cost (expressed in bits) associated with a container grouping $G$ is calculated as
\begin{equation}
\label{eq:storageCostGrouping}
storageCost(G) = \displaystyle \sum_{S \in G} storageCost(S)\enspace ,
\end{equation} 
\noindent namely, it is the sum of the storage costs of each container subset $S$ contained within $G$.
\end{definition}

\subsection{Modeling Compression Gain}
\label{subsec:compGainModel}
Two distinct gain measures are associated with each container grouping: the {\em local compression gain (localGain)} indicates the compression gain obtained by using the current grouping, while the {\em maximum potential compression gain (mpGain)} indicates the highest possible compression gain that can be obtained moving forward by chosing any partitioning strategy that ``agrees with'' the current grouping (i.e.,~there exists no container $C$ such that the current grouping and the partitioning strategy place $C$ within different container subsets). Both measures are used in the first phase of the algorithm to guide the search for candidate container partitioning strategies, and we presently describe how both measures are calculated.

\begin{definition} The {\em local compression gain} (expressed in bits) of a container subset $S$, denoted {\em localGain($S$)}, is calculated as 
\begin{equation}
\label{eq:localSubsetGain}
localGain(S) = \max \{0, \Gamma(S)\}\enspace ,
\end{equation}
where

\begin{equation}
\label{eq:positiveSubsetGain}
\Gamma(S) = 8 \cdot |S| - (C_{\mathrm{LZ}}(S) \cdot |S| + storageCost(S))
\end{equation}
and $|S|$ indicates the total byte length of the contents of $S$.
\end{definition}

Eq.~\ref{eq:localSubsetGain} ensures that compression is only applied if it results in a positive compression gain; otherwise, the subset $S$ is left {\em uncompressed}, and $localGain(S) = 0$.  In Eq.~\ref{eq:positiveSubsetGain}, the sum of the estimated compressed size of $S$ and the associated storage cost is subtracted from the original bit length of $S$.  This quantity represents the total number of bits saved by applying compression to $S$. Note that while Eq.~(\ref{eq:positiveSubsetGain}) assumes a byte-level compression of container contents, text encoding schemes using multiple bytes per character (e.g.,~Unicode formats) may be supported by considering each byte as an individual token.

\begin{definition}
The local compression gain of a container grouping $G$ is calculated as

\begin{equation}
localGain(G) = \displaystyle \sum_{S \in G} localGain(S)\enspace .
\end{equation}
\end{definition}

Hence, the overall local compression gain for an entire container grouping (i.e.~a set of container subsets) is simply calculated as the sum of local gains for each container subset within that grouping.

\begin{example}
Recalling the example grouping $S = \{aaabc\}$ from Ex.~\ref{ex:LZparsing}, $\Gamma(S) = 5 \cdot 8 - (0.8 \cdot 5 + (4 \cdot (8 + \log2(4)))) = -4$ bits and therefore $localGain(S) = 0$ bits, indicating that $S$ should be left uncompressed.  
\end{example}

\linesnumbered
\begin{algorithm}[t]
\KwIn{$D$, the set of existing LZ76 dictionaries for the grouping $G$; $c_t$, total number of characters in all containers of $\mathcal{C}$; $c_u$, number of remaining unprocessed characters.}
\KwOut{$mpGain(G)$, indicating the maximum potential compression gain for $G$.}
\begin{enumerate}
\item Choose the dictionary $d \in D$ containing the phrase of longest length, and let $S_{max}$ denote the container subset whose dictionary is $d$.  In case of a tie, choose the subset with the lowest $C_{\mathrm{LZ}}$ value. Set $nPhrases$ to be the number of phrase entries in $d$, and $maxPhraseLength$ to be the length of the longest phrase, plus one.
\item While $c_u \ge maxPhraseLength$, simulate the creation of a new, longer phrase by performing the following steps:
\begin{enumerate}
\item Set $c_u = c_u - maxPhraseLength$. This reduces the number of unprocessed characters to include only those not covered by the new phrase.
\item Set $nPhrases = nPhrases + 1$, to update the count of phrases in the dictionary $d$.
\item Set $maxPhraseLength = maxPhraseLength + 1$, to ensure the next created phrase (if applicable) will have a length one character longer than the current longest phrase.
\end{enumerate}
\item At this point, if any unprocessed characters remain (i.e.,~$c_u > 0$), this number is less than the length of the next new phrase to be created.  To handle the remaining characters, we just choose an existing phrase of length $c_u$ from $d$ and no additional phrases will be added from this point. 
\item Compute $C_{\mathrm{LZ}}(S_{max}) = \frac{nPhrases}{c_t}$, and use this value to recalculate $localGain(S_{max})$. 
\item Compute and return $mpGain(G) = \displaystyle \sum_{S \in G \setminus S_{max}} localGain(S) + localGain(S_{max})$.
\end{enumerate}
\caption{Calculation of maximum potential compression gain.}
\label{alg:maxGain}
\end{algorithm}

As mentioned above, the maximum potential compression gain is used to indicate the upper bound on the achievable compression gain for any partitioning strategy that agrees with the current grouping.  Since the total number of characters (i.e.,~the number of characters contained within the existing grouping $G$, plus the number of characters contained within containers that have yet to be assigned to subsets) is fixed, so too is the first product in Eq.~\ref{eq:positiveSubsetGain}, and maximizing compression gain over a subset $S$ then requires the sum of $C_{\mathrm{LZ}}(S)$ and $storageCost(S)$ to be minimized. From Eq.~\ref{eq:LZ} and Eq.~\ref{eq:storageCost}, one observes that both quantities are minimized when the number of generated phrases is also minimized.  Equivalently, at each step during LZ76 parsing, one seeks to generate the {\em longest applicable phrase} by appending an extra character to the longest existing phrase in the dictionary. Alg.~\ref{alg:maxGain} illustrates how the maximum potential gain is calculated for a grouping.   

In the first step, the longest phrase over all subset dictionaries is identified. For the container subset $S_{max}$ whose dictionary contains this longest phrase, the existing dictionary is extended with longer phrases, until no unprocessed characters remain.  More precisely, each iteration of step 2 creates a new phrase one character longer than the previous longest phrase, and applies it to the sequence of unprocessed characters. Eventually, either all remaining characters will be processed, or the number of remaining characters will be less than the longest phrase.  In the latter case, a shorter existing phrase is reused to cover the remaining characters (step 3).  Step 4 computes the new value of $C_{\mathrm{LZ}}(S_{max})$, and updates the value of $localGain(S_{max})$.  Finally, step 5 computes the $mpGain$ for the grouping $G$ by summing the updated $localGain$ score for $S_{max}$ with the existing $localGain$ scores for the remaining subsets in $G$.

\begin{example}
To illustrate the computation of $mpGain$, we recall from Ex.~\ref{ex:LZparsing} the previous example subset $S = \{aaabc\}$, and the dictionary of phrases $\{\langle a\rangle,\langle aa\rangle,\langle b\rangle,\langle c\rangle\}$ that results from an LZ76 parsing of $S$.  Assume that there is one additional container $C_x$ with 5 characters.  Alg.~\ref{alg:maxGain} first selects the longest existing phrase $\langle aa\rangle$ and constructs a new phrase of length 3.  Applying this to $C_x$ leaves only $5 - 3 = 2$ remaining unprocessed characters, a number which is less than 3, the current maximum phrase length.  Therefore, the existing pattern $\langle aa\rangle$ is applied, and no unprocessed characters remain. Only one additional pattern has been created, and the new complexity score is $5/10 = 0.5$ symbols per character, while the updated storage cost is $5 \cdot (8 + \log_2(5)) \approx 51.6096$ bits, and $mpGain(S) \approx 10 \cdot 8 - (0.5 \cdot 10 + 51.6096) \approx 23.3904$ bits.   
\end{example}

\subsection{Branch-and-Bound Algorithm for Selecting Candidate Partitioning Strategies}
\label{subsec:BandB}
In this phase, a search tree is constructed in which each node corresponds to a particular grouping.  Each node stores the $localGain$ and $mpGain$ values for its associated grouping.  The subtree rooted by a node $n$ encompasses all groupings that extend the grouping associated with $n$ by assigning additional containers to container subsets.  

Before explaining the details of the branch-and-bound procedure, we begin with an intuition as to why this technique is applicable to the subproblem of choosing a container grouping. Recall that the $mpGain$ indicates the highest possible gain possible for any partitioning strategy based on the current grouping. In addition, we may also observe that $mpGain(p) \ge mpGain(c)$ for any parent node $p$ and child node $c$ in the search tree.  This is due to the fact that there are fewer remaining unprocessed characters as one travels from $p$ to $c$: in particular, the placement of one additional container has been ``fixed'' by the grouping associated with $c$.  At the lowest level of the search tree, {\em all} containers have a fixed placement (i.e., each leaf node corresponds to a partitioning strategy), and therefore $mpGain$ will equal $localGain$ for each leaf node.  

Exploiting these properties of the $mpGain$ measure provides us with our {\em bounding criterion}: if the $mpGain$ for a grouping is sufficiently less than the best local gain value encountered thus far, the entire subtree rooted at the node representing the grouping can be immediately eliminated from consideration (or ``killed''). We now are in a position to describe the specifics of the branch-and-bound procedure.

The inputs to the procedure are a set of containers $\mathcal{C}$, sorted in descending order of their respective sizes, along with an additional parameter $\delta \in \mathbb{R}^{+}$. The latter specifies a threshold value used to determine whether a particular node should be ``killed'', or if it is worthwhile to continue branching into its subtree (in which case it is considered to be a ``live'' node). During the search procedure, the optimal local gain value encountered so far is stored in variable $optGain$. The root node of the search tree is assigned the grouping $\{C_1\}$, that is, a single set containing only the first container.  For $i = 2,...,|\mathcal{C}|$, the steps in Alg.~\ref{alg:BandB} are carried out to enumerate the various choices for placement of each container $C_i$ within the context of an existing grouping (where each such choice corresponds to a child node of the existing grouping node), and to determine the optimal choice of placement among the alternatives. Note that in Alg.~\ref{alg:BandB}, $x$ refers to the node currently being evaluated in the tree, and $G_x$ refers to the container grouping associated with $x$.  

For each ``live'' node $p$ at level $i$ in the tree, a set of child nodes are constructed; each represents a different strategy for placing the container $C_{i+1}$ into either a new subset, or within one of the existing container subsets present in the grouping associated with $p$.  Once all ``live'' nodes at level $i$ have been branched, $mpGain$ and $localGain$ values for all nodes at level $i$ are computed, and if necessary $optGain$ is updated to reflect a new global maximum for $localGain$.  For each node having a $localGain$ less than $optGain$, a test is carried out to ensure that its $mpGain$ falls within the range $[optGain - \delta,optGain]$.  If the test fails, the node is ``killed''.  Further branching is only carried out at level $i+1$ on the remaining live nodes at level $i$; at each iteration, the unbranched node at level $i$ with the highest $mpGain$ value is chosen. At level $|\mathcal{C}|$, the remaining live nodes will comprise the set $\mathcal{G}$ of candidate partitioning strategies.

\begin{algorithm}[t]
\KwIn{container $C_i \in \mathcal{C}$, context node $x$, and a threshold value $\delta \in \mathbb{R}^{+}$}
\KwOut{a set $\mathcal{G}$ of candidate container partitioning strategies}
\begin{enumerate}
\item Construct as the leftmost child of $x$ the grouping formed by adding the single-container subset $\{C_i\}$ to $G_x$. 
\item For each existing subset $S \in G_x$, add a child to $x$ corresponding to the grouping formed by $G_x \setminus \{S\} \cup \{S \cup C_i\}$. 
\item For each of the child nodes $y$ created in steps 1 and 2, let $G_y$ represent the grouping associated with $y$ and calculate $localGain(G_y)$ and $mpGain(G_y)$.
\item If one of the newly constructed children nodes $y$ results in a $localGain(G_y)$ value that is higher than $optGain$, set $optGain$ to this value.
\item ``Kill'' any child nodes $y$ for which $mpGain(G_y) < optGain - \delta$.
\end{enumerate}
\caption{Construction of the branch-and-bound search tree.}
\label{alg:BandB}
\end{algorithm}

We illustrate the working of the branch-and-bound procedure with the following example.

\begin{example}
\label{ex:bAndB}
Assume that we have the container set $\mathcal{C} = \{C_1, C_2, C_3\}$, where the respective contents of the containers are $C_1 = \{aaabcaaabcaaabcabcab\}$, $C_2 = \{15720653197608243849\}$, and $C_3 = \{abcababcbaaaabcabcab\}$.  We set $\delta = 30.0$ bits. Fig.~\ref{fig:BranchBound} depicts the search tree formed by this process.  Initially, the grouping $\{C_1\}$ is formed as the root of the search tree.  

In the second level of the tree, both possibilities for incorporating container $C_2$ into the existing grouping are considered: either creating a second subset to store $C_2$, or combining $C_2$ with $C_1$ in a single container subset.  The left child of the root corresponds to the first choice, creating the grouping $\{C_1\},\{C_2\}$, while the right child represents the second choice as the grouping $\{C_1,C_2\}$.  The gain values are then calculated for both children, and the test in Step 5 is performed which determines that neither node can be ``killed''.  Therefore, child nodes are constructed for both nodes based on the possibilities for assigning container $C_3$.  

In the case of the left child, there are three possibilities: assigning $C_3$ to a third, separate subset to form the grouping $\{C_1\},\{C_2\},\{C_3\}$; appending $C_3$ to the first existing subset to create the grouping $\{C_1,C_3\},\{C_2\}$; or adding $C_3$ to the second existing subset, creating the grouping $\{C_1\},\{C_2,C_3\}$.  Hence, three children are created within the left subtree.  A similar process creates two child nodes in the right subtree, corresponding to the choice of creating a new subset for $C_3$ (generating the grouping $\{C_1,C_2\},\{C_3\}$), or combining it within the existing subset of the parent grouping (forming the grouping $\{C_1,C_2,C_3\}$.  The best local gain is achieved by the grouping $\{C_1,C_3\},\{C_2\}$; when we compare the $mpGains$ of the other nodes at level 3, only $\{C_1\},\{C_2\},\{C_3\}$ comes within $\delta = 30.0$ bits of this optimal local gain.  Hence, only these two nodes remain alive, and the other three are ``killed''.  Since all three containers have now been assigned, we return the two remaining live nodes at level three as the set of candidate partitioning strategies, $\mathcal{G}$.       
\end{example}

\begin{figure}[t]
\centering
\includegraphics[scale=0.55]{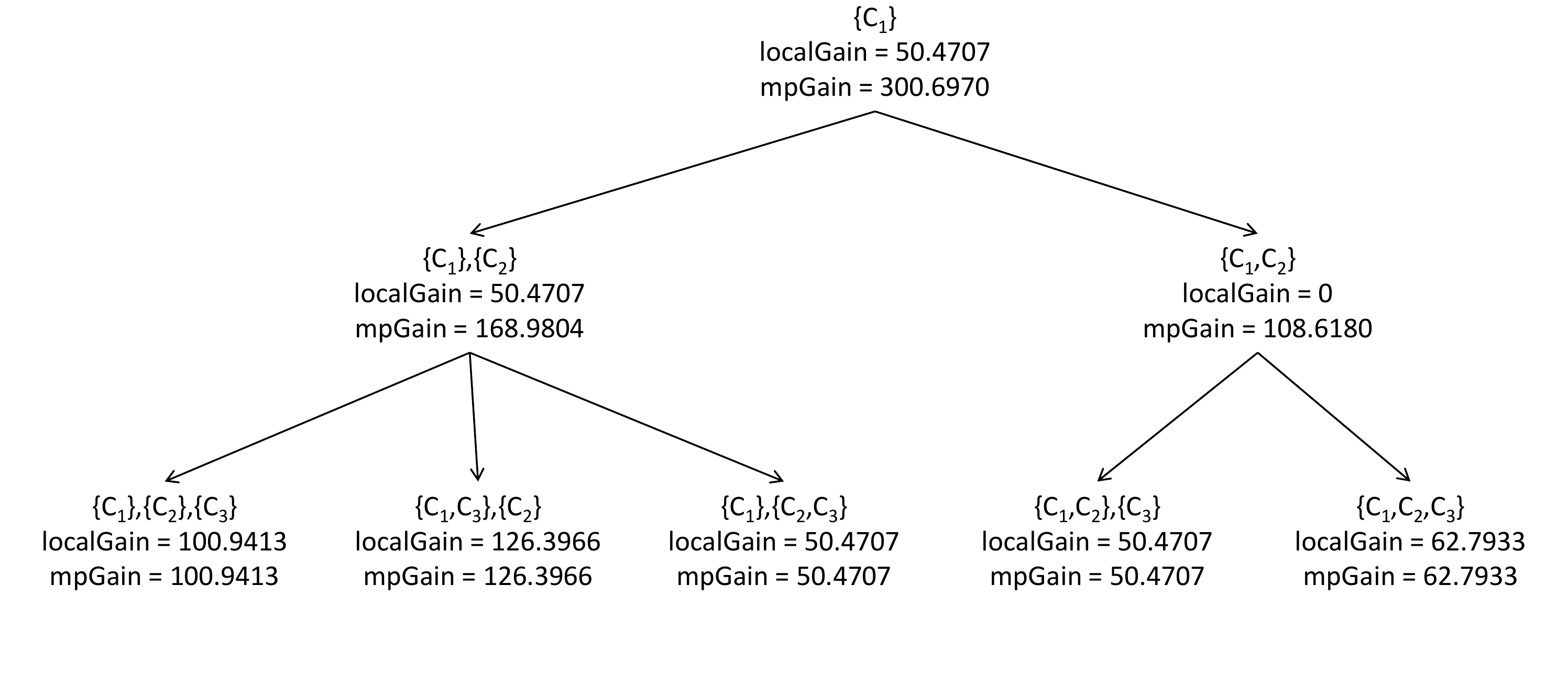}
\caption{Branch-and-bound search tree for Ex.~\ref{ex:bAndB}.}
\label{fig:BranchBound}
\end{figure}

The pruning criterion in step 5 serves to reduce the size of the search space, yet it is crucial to ensure that it does not result in the removal of the node with the highest local compression gain (the optimal node).  The following result proves that the optimum node will never be ``killed''.

\begin{figure}[t]
\centering
\includegraphics[scale=0.55]{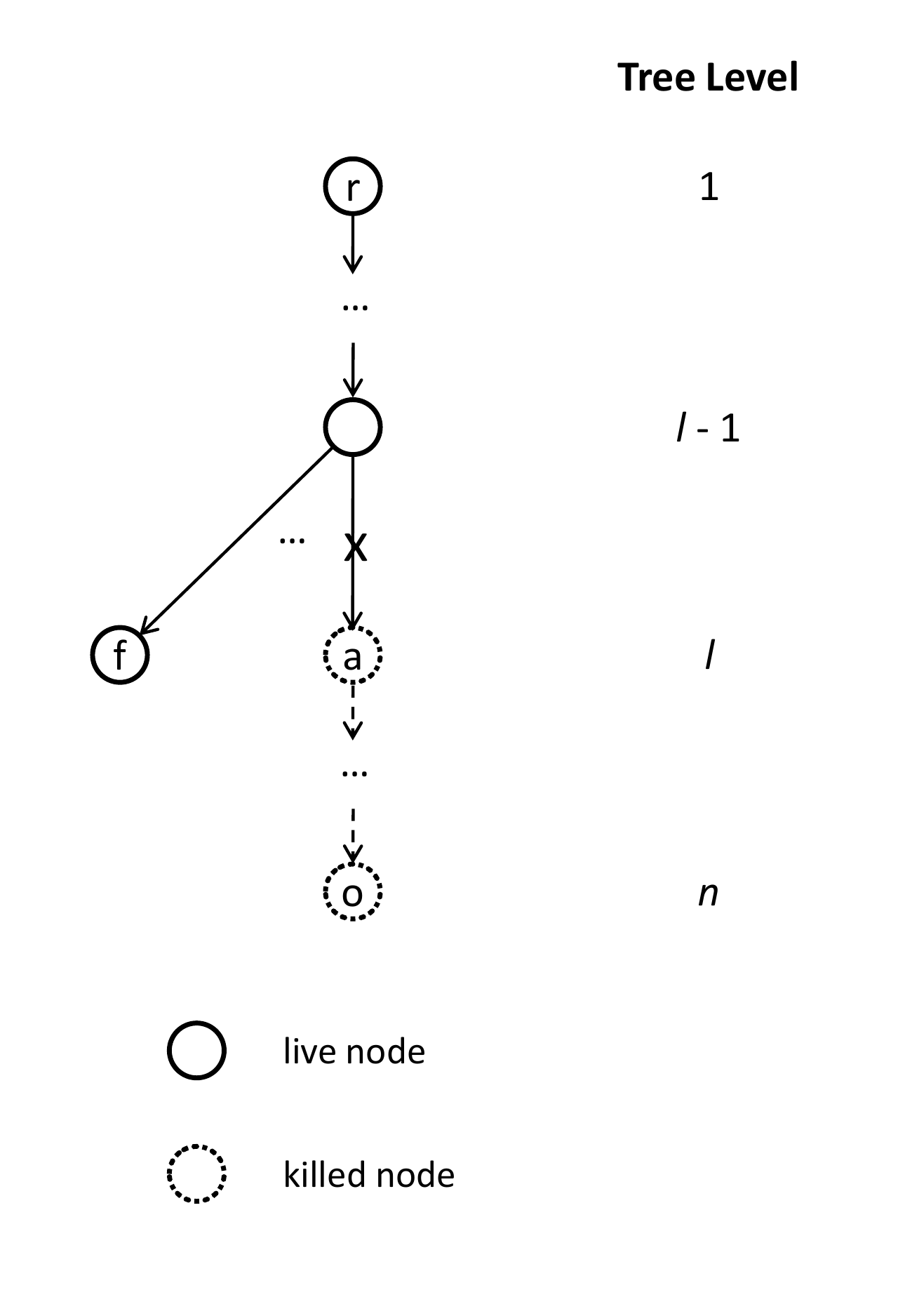}
\caption{Illustrating the proof of Proposition~\ref{prop:optimalityGuarantee}.}
\label{fig:BranchBound}
\end{figure}

\begin{proposition}
\label{prop:optimalityGuarantee}
Alg.~\ref{alg:BandB} ensures that the optimal node is always visited.
\end{proposition} 
\begin{proof}(by contradiction)
Assume that the supposed optimal node $o$ occurs at level $n$ in the tree.  Let $a$ be the highest ancestor of $o$ that has been ``killed'' by the pruning criterion of step 5, let $f$ be the (supposedly) false optimal node, and let $l < n$ be the level at which $a$ and $f$ occur.  Fig.~\ref{fig:BranchBound} provides an illustration of this scenario.  We first prove that for arbitrary $a$ and $o$, $mpGain(a) \ge localGain(o)$.  Assume instead that $localGain(o) > mpGain(a)$.  Since the groupings represented by $a$ and $o$ agree on the placement of the first $l$ containers, there must be at least one level $k$, for $k \in [l+1,n]$, at which the placement of container $C_k$ under the grouping associated with $o$ yields a higher gain than the placement specified by the grouping associated with $a$.  Without loss of generality, we first assume that $a$ is the direct ancestor (parent) of $o$, and hence $a$ and $o$ agree on the placement of the first $n-1$ containers, and only disagree on the choice of placement for container $C_n$.  Then

\begin{equation*}
\begin{split}   
localGain(o) &= \displaystyle \sum_{S_{i} \in o \wedge C_n \notin S_i} localGain(S_i) + \displaystyle \sum_{S_{j} \in o \wedge C_n \in S_j} localGain(S_j) \\
 &> \\
mpGain(a) &= \displaystyle \sum_{S_{i} \in a \wedge C_n \notin S_i} localGain(S_i) + \displaystyle \sum_{S_{j}^{\prime} \in a \wedge C_n \in S_{j}^{\prime}} localGain(S_{j}^{\prime})
\end{split}
\end{equation*}
where $S_j$ is the subset containing container $C_n$ under grouping $o$, and $S_{j}^{\prime}$ is the chosen subset for $C_n$ chosen by Alg.~\ref{alg:maxGain} for the input grouping $a$.  Simplifying, we obtain

\begin{equation*}
localGain(S_j) > localGain(S_{j}^{\prime}) 
\end{equation*}
which contradicts Alg.~\ref{alg:maxGain}. Recall that at each step, Alg.~\ref{alg:maxGain} applies the longest possible phrase, guaranteeing that no other LZ76 parsing strategy could yield a lower complexity score $C_{\mathrm{LZ}}$. This result can be generalized to the case where $a$ is an {\em indirect} ancestor of $o$, since $mpGain$ increases monotonically as one travels upward in a subtree.

Returning to the original claim, we note that if $o$ is optimal, then it must be greater than the value of $optGain_{n-1}$ (i.e.,~the globally optimal $localGain$ witnessed after $n-1$ levels of the search tree have been processed). Additionally, since the ancestor node $a$ has been ``killed'' during Step 5 of the branch-and-bound algorithm, it must also be true that $mpGain(a) < optGain_{l} - \delta$ (here, $optGain_{l}$ denotes the global optimum seen after $l$ levels of the search tree have been visited).  Since $optGain_{n-1} \ge optGain_{l}$, the only means of satisfying both conditions is for the following chain of inequalities to be satisfied:

$$localGain(o) > optGain_{n-1} \ge optGain_{l} \ge mpGain(a) + \delta\,.$$

This contradicts the previous result indicating that $mpGain(a) \ge localGain(o)$, completing the proof.\qed
\end{proof}

\subsection{Determining an Optimal Compression Configuration}
\label{subsec:optConfig}

Alg.~\ref{alg:configurationSelection} allows one to determine an optimal compression configuration from an input set $\mathcal{G}$ of candidate partitioning strategies (obtained from Alg.~\ref{alg:BandB}) and set $\mathcal{A}$ of compression algorithms, together with upper bounds on compression and decompression time, $T_c$ and $T_d$.  The variable $globalBestGain$ records the highest overall compression gain from the partitioning strategy/algorithm assignment combinations tested so far. Lines 2-29 iterate through each candidate partitioning strategy $G \in \mathcal{G}$; each container subset $S$ contained in $G$ is tested (Lines 4-23) to determine the compression algorithm $a \in \mathcal{A}$ that achieves the highest compression gain (Lines 6-14). Before an algorithm is assigned to a container subset, a test is performed to ensure that the required compression and decompression time values fall below the respective bounds $T_c$ and $T_d$ (Line 8).  

At the conclusion of testing, if there is no available algorithm in $\mathcal{A}$ that satisfies the time bounds for compression and decompression when applied to a specific subset $S$, the entire partitioning strategy containing $S$ is immediately disqualified (Lines 15-16).  Otherwise, the overall compression gain and compression/decompression time scores are updated for the partitioning strategy $G$, and the appropriate compression algorithm is assigned to the active subset $S$ (Lines 17-22).  After each partitioning strategy $G$ has been processed, a test is done to determine whether it yields a better gain than the current $globalBestGain$; if necessary, the globally-best compression configuration $\langle P, arg \rangle$ is updated to store the current partitioning strategy $G$, along with the optimal algorithm selection strategy $\alpha_{G}$ found for $G$ (Lines 24-28).  

After all partitions in $\mathcal{G}$ have been processed, the optimal compression configuration $\langle P, \alpha\rangle$ is returned (Line 30).    

\linesnumbered
\begin{algorithm}[t]
\SetKw{KwGoTo}{goto}
\small
\KwIn{set of compression algorithms $\mathcal{A}$, set of candidate container partitions $\mathcal{G}$, upper bound $T_c \in \mathbb{Z}^{+}$ on compression time, upper bound $T_d \in \mathbb{Z}^{+}$ on decompression time}
\KwOut{a compression configuration $\langle P, \alpha\rangle$}
$globalBestGain \leftarrow 0$; $P \leftarrow NULL$; $alg \leftarrow NULL$\;  
\ForEach{$G \in \mathcal{G}$}{
	$groupingCTime \leftarrow 0$;
	$groupingDTime \leftarrow 0$;
	$groupingGain \leftarrow 0$\;
	\ForEach{$S \in G$}{
		$maxGain \leftarrow 0$;
		$bestCTime \leftarrow 0$;
		$bestDTime \leftarrow 0$;
		$bestAlgorithm \leftarrow NULL$\;
		\ForEach{$a \in \mathcal{A}$}{
			$gain \leftarrow compressedGain(S,a)$\;
			\If{$gain > maxGain$ {\bf and} $groupingCTime + compressTime(S,a) \le T_c$ {\bf and} \hskip0.2em $groupingDTime + decompressTime(S,a) \le T_d$}{
				$bestCTime \leftarrow compressTime(S,a)$\;
				$bestDTime \leftarrow decompressTime(S,a)$\;
				$bestAlgorithm \leftarrow a$\;
				$maxGain \leftarrow gain$\;			
			}
		}
		\eIf{$bestAlgorithm = NULL$}{
			\KwGoTo line 4\;
		}{
			$groupingCTime \leftarrow groupingCTime + bestCTime$\;
			$groupingDTime \leftarrow groupingDTime + bestDTime$\;
			$groupingGain \leftarrow groupingGain + maxGain$\;
			$\alpha_{G}(S) \leftarrow bestAlgorithm$\;
		}
	}
	
	\If{$groupingGain > globalBestGain$}{
		$globalBestGain \leftarrow groupingGain$\;
		$P \leftarrow G$\;
		$\alpha \leftarrow \alpha_{G}$\;
	}
}

\Return{$\langle P,\alpha\rangle$\;}
\caption{Selecting a compression configuration.}
\label{alg:configurationSelection}
\end{algorithm}

\subsection{Discussion \& Practical Considerations}

\paragraph{Comparison with greedy algorithm.} At first glance, it might seem tempting to employ a simple greedy strategy for selecting a partitioning strategy.  In terms of the branch-and-bound search tree, this corresponds to exploring only the root-to-leaf branch formed by selecting at each level the child node with the highest local gain value.  While such a strategy is more time efficient, it is not too difficult to envision scenarios where it results in a sub-optimal partitioning strategy being chosen.  As an example, we can consider an XML document consisting of DBLP-like bibliographic entries.  At an earlier stage, both the branch-and-bound and greedy strategies may decide to place ``author'' and ``year'' containers in different subsets due to their dissimilarities. At a later stage, suppose that the container storing ``key'' values for entries is to be assigned, and further, assume that ``key'' values are formed by concatenating the author name with the year.  This quite possibly causes a higher compression gain to be obtainable by grouping ``author'', ``year'', and ``key'' containers together.  While the branch-and-bound strategy is capable of reconsidering the initial decision to separate ``author'' and ``year'' containers after observing the characteristics of future containers, the greedy strategy is not.   

\paragraph{Choosing a good $\delta$ value.} Choosing a value for $\delta$ represents a tradeoff between accuracy and running time.  Smaller $\delta$ values will cause more nodes to be ``killed'' at each level in the search tree, reducing the size of the tree which must be explored.  On the other hand, this also increases the likelihood that the true optimal configuration will not be chosen: a partitioning strategy may be discarded in phase one of the algorithm based solely on having a lower calculated local compression gain value, even though it may outperform all of the chosen candidate partitioning strategies in the second phase when a particular algorithm selection strategy is applied to it.  Conversely, choosing a large enough $\delta$ value ensures that the optimal configuration is always chosen, at the potential expense of an exhaustive enumeration of all possible combinations of container groupings and algorithm selections, requiring time exponential in $|\mathcal{C}|$.  

\paragraph{Ordering of containers.} Proposition~\ref{prop:optimalityGuarantee} illustrates that the optimal container partitioning strategy can never be ``killed'', implying that the ordering of containers does not impact the discovery of the optimal grouping.  Yet container ordering can affect the efficiency of the algorithm's first phase. In particular, sorting the containers in descending order of their sizes can dramatically reduce the number of tree nodes that are explored. Such an ordering causes a larger number of characters to have fixed container assigments at an earlier level in the tree, thereby reducing the number of unprocessed characters and allowing tighter bounds on $mpGain$ to be established more quickly.  As a result, larger subtrees can be pruned from the search tree.

\section{Conclusion}
\label{sec:conclusion}
In this paper, we demonstrated that determining an optimal configuration for permutation-based XML compression is an {\bf NP}-hard problem.  We also described an approximation algorithm that allows one, with proper selection of parameter values, to discover the optimal compression configuration in polynomial time (w.r.t. the sizes of the document and the set of compression algorithms $\mathcal{A}$).  As future work, we plan to implement this algorithm within our existing XML-conscious compressor~\cite{Leighton05a} and test its effectiveness via experimentation over a range of real-world and synthetic XML documents.

\bibliographystyle{splncs}
\bibliography{main}
\end{document}